\documentclass[authoryear,preprint,12pt]{elsarticle}
\usepackage{natbib}
\usepackage{amsthm}
\usepackage{amssymb}
\usepackage{graphicx,amsmath} 
\usepackage{tikz}
\usepackage{natbib}
\usepackage{booktabs,xcolor}
\usepackage{graphics}
\usepackage{enumitem}
\usepackage{subcaption}
\usepackage{caption}
\usepackage{setspace}
\usepackage[colorlinks=true,linkcolor=blue, urlcolor=blue, citecolor=blue, breaklinks]{hyperref}
\journal{ }
\newtheorem{theorem}{Theorem}

\newtheorem{lemma}{Lemma}
\theoremstyle{definition}
\newtheorem{example}{Example}
\newtheorem{definition}{Definition}

\newcommand{\RR}{\mathbb{R}}
\newcommand{\PP}{\mathcal{P}}
\newcommand{\OO}{\mathcal{O}}
\newcommand{\ZZ}{\mathcal{Z}}
\newcommand{\UU}{\mathcal{U}}
\newcommand{\FF}{\mathcal{F}}
\newcommand{\es}{\text{ES}} 
\newcommand{\IR}{\mathcal{I}}
\newcommand{\CC}{\mathcal{C}} 
\usepackage{etoolbox}
\makeatletter
\patchcmd{\ps@pprintTitle}% <cmd>
{Preprint submitted to}% <search>
{}% <replace>
{}{}% <succes><failure>
\makeatother
\patchcmd{\emailauthor}{(#2)}{}{}{}
\patchcmd{\urlauthor}{(#2)}{}{}{}

\begin{document}
\begin{frontmatter}

\title{Multi-unit Assignment under Dichotomous Preferences}
\author{Josue Ortega\corref{mycorrespondingauthor}\footnote{This paper is based on the second chapter of my PhD thesis supervised by Herv\'e Moulin, to whom I am indebted for his valuable suggestions. I have also received helpful comments from Anna Bogomolnaia, Yeon-Koo Che, Bram Driesen, Georg N\"{o}ldeke, Fedor Sandomirskiy, Erel Segal-Halevi, Jay Sethuraman, Olivier Tercieux, and audiences at the Winter Meeting of the Econometric Society in Barcelona and the Tenth Conference on Economic Design at York. Special thanks to Columbia University for their kind hospitality. I thank Erin and Christina Loennblad for proofreading the paper. Errors are mine alone.}}
\ead{josue.ortega@essex.ac.uk}
\address{University of Essex, Colchester, UK.\\
Center for European Economic Research (ZEW), Mannheim, Germany.}

\begin{abstract}
I study the problem of allocating objects among agents without using money. Agents can receive several objects and have dichotomous preferences, meaning that they either consider objects to be acceptable or not. In this set-up, the egalitarian solution is more appealing than the competitive equilibrium with equal incomes because it is Lorenz dominant, unique in utilities, and group strategy-proof. Moreover, it can be adapted to satisfy a new fairness axiom that arises naturally in this context. Both solutions are disjoint.
\end{abstract}

\end{frontmatter}

\newpage
\section{Introduction}
An assignment problem is an allocation problem where scarce objects are to be allocated among several agents without using monetary transfers. Assignment problems include the allocation of senators to committees, courses to students, or job interviews to applicants. In this paper, I study assignment problems in which each agent can receive more than one object, but at most one unit of each, and several identical units are available of each object. These are called multi-unit assignment problems. They include the three examples previously discussed. A U.S. senator on average participates in four committees,\footnote{Source: \href{https://corg.indiana.edu/many-roles-member-congress}{``The many roles of a Member of Congress''}, {\it Indiana University Center on Representative Government}.} a student can take many courses during a semester, and a job candidate can schedule many interviews. However, senators cannot have more than one seat on each committee, students cannot take a course twice for credit, and applicants cannot be interviewed more than once for the same position. 

For such multi-unit assignment problems, we would like to have a systematic (probabilistic) procedure to decide fairly which agents should get which objects, which, at the same time, does not offer incentives to coalitions of agents to lie about their true preferences. My contribution is to propose an egalitarian solution that achieves this purpose for multi-unit assignment problems in the dichotomous preference domain, in which objects are either considered acceptable or not, and in which agents are indifferent between all objects that they find acceptable. 

The egalitarian solution is based on the well-known leximin principle. In the domain of dichotomous preferences, it performs better than the celebrated competitive equilibrium with equal incomes, a solution used in similar assignment models on larger preference domains \citep{hylland1979,budish2011, reny2017}, and which has been successfully applied to allocate courses in business schools \citep{budish2017}. By better, I mean that, unlike the competitive equilibrium with equal incomes, the egalitarian solution is Lorenz dominant, unique in utilities, and impossible to manipulate by groups. In contrast with the single-unit assignment problem \citep{bogomolnaia2004}, both solutions are disjoint, meaning that in general we cannot obtain the \emph{egalitarian} solution as a competitive equilibrium when agents are endowed with \emph{equal incomes}. 

Lorenz dominance is ``{\it a ranking generally accepted as the unambiguous arbiter of inequality comparison}'' \citep{foster1999} and is ``{\it widely accepted as embodying a set of minimal ethical judgements that should be made}'' \citep{dutta1989}. Given two vectors of size $n$, the first Lorenz dominates the second if, when arranged in ascending order, the sum of the first $k\leq n$ elements of the first is always greater than or equal to the sum of the $k$ first elements of the second. A utility profile is Lorenz dominant if it Lorenz dominates any other feasible utility profile. In our set-up, the fact that a utility profile is Lorenz dominant implies that it uniquely maximizes any strictly concave utility function representing agents' preferences and is, therefore, a strong fairness property.

Uniqueness of the solution (in the utility profile obtained) is also a desirable property, for it gives a clear recommendation of how the resources should be split. A multi-valued solution leaves the schedule designer with the complicated task of selecting a particular division among those suggested by the solution, thus raising the possibility of justified complaints by some agents who may argue that other allocations were also recommended by the solution that were more beneficial to them.

It is equally interesting that the egalitarian solution is group strategy-proof, implying that coalitions of agents can never profit from misrepresenting their availability. On the contrary, the competitive solution is manipulable by groups in this set-up, as in many others. Yet, it is remarkable that even in our small dichotomous preference domain, where the possibilities to misreport are very limited, the pseudo-market solution can still be manipulated by coalitions of agents.

The fact that the egalitarian solution satisfies these three desirable properties is a strong argument for recommending its use whenever agents have dichotomous preferences, instead of the competitive equilibrium with equal incomes. 

The dichotomous preference domain is admittedly simple, and is not suitable for modeling some multi-unit assignment problems in which agents may consider objects as complements, such as the allocation of courses to students. However, this set-up is helpful to represent scheduling problems (see for example the tennis allocation problem in \citealp{maher2016}), in which agents are either compatible or incompatible with each object and want to maximize the number of objects they obtain, or for the aforementioned problems of assigning job interviews to candidates or seats for performances to the public, among others.

Moreover, focusing on this particular domain of preferences will be helpful to show the properties of the egalitarian solution, while, at the same time, it will make the problem complicated enough to identify why the competitive equilibrium with equal incomes fails to be unique and group strategy-proof. The reason behind its non-uniqueness is that for some objects, the number of identical copies available of them (their supply) equals their total demand. I call these objects perfect. Although there is no doubt on how perfect objects should be allocated, the question of how to price them becomes tricky. Because their demand is always equal to their supply, they can have a zero competitive price. However, they could also have a positive price, hence reducing the available budget of those agents who buy them. 

Perfect objects are also the reason why the competitive solution is not group strategy-proof. A coalition of agents can agree to misrepresent their preferences in order to make a set of objects perfect. This allows those objects to have a lower price (even a zero price), thus	 allowing agents to increase their budget, and, consequently, their share of other over-demanded objects. Manipulating agents benefit unambiguously, meaning that every competitive equilibrium of the allocation problem with misrepresented preferences yields a weakly better allocation than the unique competitive equilibrium of the original problem.

More generally, perfect objects also raise the issue of how they should affect the final allocation. Some allocation procedures can be decomposed into the allocation of perfect and over-demanded objects, meaning that the share of over-demanded objects that agents obtain is independent of their demand for perfect objects. I call this property independence of perfect objects. This is a desirable property in scenarios where agents can claim that perfect objects belong unambiguously to them and the number of perfect objects they obtain should not diminish their shares of over-demanded objects. Although the egalitarian solution does not satisfy this requirement, we can construct a refined egalitarian solution that does and is also Lorenz dominant for the assignment problem with over-demanded objects only. However, independence of perfect objects comes at a price: the refined egalitarian solution is not group strategy-proof.

\subsection{Related Literature} 
The theoretical model I study is closely related to two existing problems in the literature:

\noindent
\begin{enumerate}[leftmargin=0cm,itemindent=.5cm,labelwidth=\itemindent,labelsep=0cm,align=left]
\item {\it Single-unit random assignment with dichotomous preferences} by \cite{bogomolnaia2004}, henceforth BM04. Our model generalizes theirs in that agents can receive more than one object. They study the egalitarian and the equal income competitive solution. They show that the egalitarian solution is Lorenz dominant and can always be supported by competitive prices. Because the competitive solution is Lorenz dominant, the competitive solution can easily be computed as the maximization of the Nash product of agents' utilities. They also prove that the egalitarian solution is group strategy-proof.

\cite{roth2005} show that the egalitarian solution is also Lorenz dominant in assignment problems on arbitrary graphs that are not necessarily bipartite. They use dichotomous preferences to model whether a person is compatible with a particular organ for transplantation. In their words, {``\it the experience of American surgeons suggests that preferences over kidneys can be well approximated as 0 - 1, i.e. that patients and surgeons should be more or less indifferent among kidneys from healthy donors that are blood type and 	immunologically compatible with the patient''}.

Assignment on the dichotomous domain of preferences has been further studied by \citet{bogomolnaia2005}, \citet{katta2006}, and \citet{bouveret2008}.

\cite{kurokawa2015} also study a single-unit assignment problem in which agents can derive a utility equal to one or zero if their demand is met. This is, if an agent demands 10 objects, he obtains the same zero utility if he receives 9, 2 or 0 objects, whereas in this paper  agents' utility is linear on the goods they find acceptable: each extra unit has the same marginal utility. They consider a broader preference domain which satisfies four conditions, namely convexity, equality, shifting allocations, and optimal utilization. They show that the egalitarian solution is group strategy-proof, envy-free and unique in utilities for every multi-unit assignment problem in this domain of preferences. They show that the egalitarian solution is not Lorenz dominant in this larger domain.

They allow for non individually rational allocations. In their set-up, an agent who wants an apple but who dislikes a pear may in fact get the pear. Furthermore, they impose that the agent derives the same utility from obtaining the apple and the pear, or the apple alone. In contrast, in BM04 and the model in this paper, agents cannot receive objects they do not find acceptable (and thus they do not have preferences over them). More technically, the property of \emph{shifting allocations} does not apply in our set-up because agents may receive non individually rational allocations.

\item {\it The course allocation problem (CAP)} described by \cite{brams2001,budish2011,budish2012,kominers2010,krishna2008}; and \cite{sonmez2010}, with some important differences. First, in CAP, students may have arbitrary preferences over the set of objects, which are considerably more complex than those I study in this paper. However, reporting combinatorial preferences is infeasible for even few alternatives, and, in practice, combinatorial mechanisms never allow agents to fully report such preferences, not only because such revelation would be complicated, but also because agents may not know their preferences in such detail. Consequently, a new strand of theory has focused on allocation mechanisms with simpler preferences (including linear utilities and dichotomous preferences, e.g. \citealp{bouveret2016} and \citealp{bogomolnaia2017}). Although the dichotomous preference domain is smaller than those considered in CAP, it is not contained in any of those because CAP rules out indifferences. 

Furthermore, \cite{budish2011} only considers deterministic assignments. I instead study randomized assignments: in practice, many allocation mechanisms use some degree of randomization to achieve a higher degree of fairness.\footnote{Randomization is used to assign both permanent visas and housing subsidies in the US, or school places in the UK. Sources: \href{https://www.theguardian.com/us-news/2017/may/02/green-card-lottery-us-immigration-trump-agenda}{``A one in a million chance at a better life''}, {\it The Guardian}, 2/5/2017, \href{http://www.npr.org/2016/05/03/476559490/why-does-random-chance-decide-who-gets-housing-subsidies}{``Why does random chance decide who gets housing subsidies?''}, {\it NPR}, 3/5/2016, and \href{https://www.theguardian.com/education/2017/mar/14/school-admissions-lottery-system-brighton}{``School admissions: is a lottery a fairer system?''}, {\it The Guardian}, 14/3/2017.}
\end{enumerate}

\subsection{Summary of Results}
\label{sec:summary}

I define the egalitarian and the constrained competitive solution. The egalitarian one is Lorenz dominant in the set of efficient utility profiles (Theorem \ref{theorem:lorenz}), while the competitive one exists (Theorem \ref{theorem:existence}) but is multi-valued (Example \ref{tab:example1}). The egalitarian solution is group strategy-proof, but the competitive one is not (Theorem \ref{theorem:manipulation}). Both solutions are disjoint (Example \ref{tab:example2}), which is a stark difference between this model and BM04. 

As a consequence, the classic result stating that the competitive solution can be computed as the maximizer of the Nash product of utilities no longer holds: a result known as the Eisenberg-Gale program.  Its failure is important not only because it leaves us with no known algorithm for computing competitive equilibria, but also because the Eisenberg-Gale program is a rather robust result that applies to a large class of utility functions beyond the linear case \citep{vazirani2007} and to the division of goods and bads \citep{bogomolnaia2017}. The fact that the competitive solution is not unique is also interesting, as a unique utility profile is always obtained in Fisher markets.

I show that the egalitarian solution is not independent of perfect objects, and propose a refined egalitarian solution that achieves this property, while at the same time being Lorenz dominant for the assignment of over-demanded objects (Lemma \ref{thm:refined}). This refined solution, while appealing, violates group strategy-proofness, unlike the classical egalitarian solution (Example \ref{tab:example3}). 

This paper is structured as follows. Sections \ref{sec:model} and \ref{sec:solutions} formalize the model and the solutions I consider, respectively. Section \ref{sec:manipulation} analyses the solutions' manipulation by groups, whereas Section \ref{sec:perfect} introduces the property of independence of perfect objects. Section \ref{sec:conclusion} concludes the paper. All proofs are provided in the Appendix. 

\section{Model}
\label{sec:model}
I consider the allocation of $m$ objects (each with possibly several copies of itself) to $n$ agents. Up to $q_k$ copies of object $k \in M$ can be assigned to the set of agents $N$. I refer to the integer vector $q=(q_1,\ldots, q_m)$ as {\bf objects' capacities}.

Agents' preferences over objects are given by a $m \times n$ binary matrix $R$. Each entry $r_{ik}=1$ if agent $i$ finds object $k$ acceptable and 0 otherwise.\footnote{Depending on the application we have in mind, $R$ can also be understood to represent either allocation or physical constraints.} Slightly abusing the notation, $R_{iM}$ (resp. $R_{Nk}$) denotes both the $i$-th row (resp. $k$-th column) of $R$ and the set of objects (resp. agents) for which $r_{ik}=1$. I assume $\left\vert R_{Nk} \right\vert \geq q_k$ for each object $k$.\footnote{If an object is under-demanded, the additional copies can always be thrown away. If we eliminate this assumption, then we need to consider assignments which are not individually rational, in which agents receive objects that they do not want.}

A {\bf random assignment matrix (RAM)} for an MAP $(R,q)$ is an $m \times n$ matrix satisfying the following conditions $\forall i \in N,k \in M$
\begin{eqnarray}
\text{Feasibility}&&\begin{cases}
0 \leq z_{ik}\leq 1\\
\sum_{k \in M} z_{ik} \leq q_k\\
\end{cases}\\
\text{Individual Rationality}&&\begin{cases}
z_{ik} >0 \text{ only if } r_{ik}=1\\
\end{cases}
\end{eqnarray}

An RAM's entries indicate what probability each agent has of obtaining one unit of each object. The feasibility conditions ensure that no agent obtains more than one unit of each object, and that the total number of units assigned of each object is less than its capacity. Similarly, individual rationality guarantees that each agent only obtains shares from acceptable objects. Throughout the paper, I only consider assignments satisfying these two properties. As before, the notation $Z_{iM}$ (resp. $Z_{Nk}$) denotes both the $i$-th row (resp. $k$-th column) of $Z$ and the set of objects (resp. players) for which $z_{ik}=1$. $\FF(R,q)$ denotes the set of all RAMs for the MAP $(R,q)$. The {\bf matching size} $\nu(R,q)=\sum_{k \in M} q_{k}$ of an MAP represents the maximum number of object units that can be assigned.

Several random assignments can have the same corresponding RAM. Theorem 1 in \cite{budish2013} implies\footnote{The implication follows because the set of feasibility constraints is a hierarchy. Lemma 1 is an extension of the well-known Birkhoff-von Neumann decomposition theorem.} that

\begin{lemma}
	Any RAM can be decomposed into a convex combination of binary RAMs, and can thus be implemented.
\end{lemma}

I assume that agents are indifferent between objects that they find acceptable, and that they want to maximize the number of acceptable objects they obtain. The canonical utility function representing those preferences is
\begin{equation}
u_i(Z)=\sum_{k \in M} z_{ik} 
\end{equation}

for an arbitrary agent $i$. This function is clearly not unique but it is convenient to work with. The preference relation represented by this function is a complete order over all RAMs, and implies that an RAM $Z$ is Pareto optimal for an MAP $(R,q)$ if and only if $\sum_{i \in N} \sum_{k \in M} z_{ik}= \nu(R,q)$. The set of efficient utility profiles $\UU(R,q)$ can be described as
\begin{equation}
\UU(R,q)=\{U \in \RR^n \mid \exists \, Z \in \FF(R,q): U_i = \sum_{k \in M} z_{ik}, \; \forall i \in N \}
\end{equation}

I do not distinguish between ex-ante and ex-post efficiency because in the dichotomous preference domain they coincide. This equivalence occurs because the sum of utilities is constant in all efficient assignments.\footnote{Ex-ante and ex-post efficiency are equivalent in assignment problems with dichotomous preferences (BM04, \citealp{roth2005}).} In our set-up, efficiency simply requires that no object is wasted.

A {\bf welfarist solution} is a mapping $\Phi$ from $(R,q)$ to a set of efficient utility profiles in $\UU(R,q)$, and hence, it only focuses on the expected number of objects received by an agent and not on the exact probability distribution over deterministic assignments. Whenever a solution is single-valued I instead use the notation $\phi$.

\subsection{Perfect Objects and Perfect Extensions}
We can partition the corresponding set of objects $M$ into two subsets $\PP(R,q)$ and $\OO(R,q)$, which are called {\bf perfect} and {\bf over-demanded}, respectively. The set of perfect objects is defined as
\begin{equation}
\label{eq:perfect} \PP(R,q) =\{k \in M : \left\vert R_{Nk} \right\vert = q_k\}\\
\end{equation}
The vectors $q_{\PP(R,q)}$ and $q_{\OO(R,q)}$ denote the capacities of perfect and over-demanded goods, respectively.

Given a MAP $(R,q)$, a {\bf perfect extension} for agent $i$ represents adding an arbitrarily perfect object $k'$ that agent $i$ finds acceptable. Formally, a perfect extension for agent $i$ in a MAP ($R,q$) is a pair $([R \, R_{Nk'}], \overline q)$ where $[R \, R_{Nk'}]$ denotes the $n \times (m+1)$ juxtaposition of the two matrices and $\overline q=(q_1,\ldots,q_m, \left\vert R_{Nk'} \right\vert)$. 

\section{Three Efficient Solutions}
\label{sec:solutions}

\subsection{The Egalitarian Solution}
 An intuitive solution equalizes agents' utilities as much as possible respecting efficiency and individual rationality: this is the well-known leximin solution. I refer to it as the {\bf Egalitarian Solution (ES)}, proposed theoretically by BM04, and applied to the exchange of live donor kidneys for transplant by \cite{roth2005} and \cite{yilmaz2011}.

To define it formally, let $\succ^l$ be the well-known lexicographic order.\footnote{So that for any two vectors $U,U' \in \RR^n$, $U\succ^l U'$ only if $U_t > U'_t$ for some integer $t \leq n$, and $U_p=U'_p$ for any positive integer $p< t$.} For each $U \in \RR^n$, let $\gamma(U) \in \RR^n$ be the vector containing the same elements as $U$ but sorted in ascending order, i.e. $\gamma_1(U)\leq \ldots \leq \gamma_n(U)$. The leximin order $\succ^{LM}$ is defined by $U\succ^{LM} U'$ if and only if $\gamma(U) \succ^l \gamma(U')$. The ES is defined by
\begin{equation}
\phi^\es(R,q)= \arg \max_{\succ^{LM}} \,  \UU (R,q)
\end{equation}

The ES satisfies a strong fairness notion called {\bf Lorenz dominance}, defined as follows. Define the order $\succ^{ld}$ on $\RR^n$ so that for any two vectors $U$ and $U'$, $U\succ^{ld} U'$ only if $\sum_{i=1}^t U_i \geq \sum_{i=1}^t U'_i$ $\; \forall t\leq n$, with strict inequality for some $t$. We say that $U$ Lorenz dominates $U'$, written $U \succ^{LD} U'$, if $\gamma(U) \succ^{ld} \gamma(U')$. A vector $U \in \UU(R,q)$ is Lorenz dominant for an MAP $(R,q)$ if it Lorenz dominates any other vector in $\UU(R,q)$. 

Lorenz dominance is a partial order in $\UU(R,q)$ and therefore a Lorenz dominant utility profile need not exist. Nevertheless, the ES solution is Lorenz dominant.

\begin{theorem}
\label{theorem:lorenz}
The ES solution is Lorenz dominant in the set of efficient utility profiles.
\end{theorem}

I prove Theorem \ref{theorem:lorenz} using Theorem 3 in \citet{dutta1989}, which states that the core of every supermodular cooperative game has a Lorenz dominant element. The construction of the corresponding cooperative game can be found in the Appendix.

\subsection{The Constrained Competitive Equilibrium with Equal Incomes}

A second solution, which is substantially more complicated, requires to balance the supply and demand for goods when agents are endowed with equal budgets. These equal budgets are often normalized to one currency unit, a normalization that I also use. This solution is known as the {\bf Competitive Equilibrium with Equal Incomes (CEEI)} \citep{varian1974,hylland1979}. In MAPs, each agent can consume at most one unit of each object, hence having particular constraints on their consumption set. I use the term {\bf Constrained Competitive Equilibrium} ({\bf CCE}, still with equal incomes) from now on to make this distinction obvious. The CCE solution is different from the CEEI as defined in \citet{hylland1979} in that in our case agents never partially consume objects that have different prices (see Table 1 in their paper).\footnote{That is, if an apple pie has a higher price than a pear pie, and the agent values them equally, the agent either fully consumes the pear pie and eats some or all of the apple pie, or the agent only consumes some of the pear pie and none of the apple pie. What can never occur in a CCE is that an agent consumes some, but not all, of the pear pie and some, but not all, of the apple pie.} This distinction justifies the different terminology of CCE.

\begin{definition}
A CCE for an MAP ($R,q$) is a pair of an RAM $Z^*$ and a non-negative price vector $p^*$ such that, $\forall i \in N$, agents maximize their utilities
\begin{equation}
Z_{iM}^* \in \arg \max_{Z_{iM} \in \beta_i(p^*)} u_i(Z_{iM})
\end{equation}

where $\beta_i(p)$ is the budget set defined as $\beta_i(p)=\{Z_{iM}\mid \sum_{k \in M} z_{ik} \leq \left\vert R_{iM} \right\vert: p \cdot Z_{iM} \leq 1\}$, and the market clears, so that 
\begin{equation}
Z^* \in \FF(R,q)
\end{equation}
\end{definition}

As we shall see in Theorem \ref{theorem:existence}, the set of CCE is never empty but may be large. The optimality conditions of CCE imply
\begin{eqnarray}
\label{eq:0} k \notin \PP(R,q) &\implies& p^*_k >0\\
\label{eq:1} z^*_{ik}, z^*_{ik'} \in (0,1) &\implies& p^*_k = p^*_{k'}\\
\label{eq:2} [p^*_k < p^*_{k'}] \land [ 0<z^*_{ik'}] &\implies& z^*_{ik} = 1\\
\label{eq:3} \sum_k z^*_{ik} < \left\vert R_{iM} \right\vert &\implies& \sum_k p^*_k \cdot z^*_{ik} = 1		
\end{eqnarray}

These are the equivalent of the Fisher equations in our model, see \cite{brainard2005}. Condition (\ref{eq:0}) allows a zero price only for perfect objects, while expression (\ref{eq:1}) forces the same marginal benefit for every object that agents obtain partially but not fully.

The CCE is in general multivalued. Given an MAP, I denote the set of pairs ($Z^*,p^*$) as $\CC(R,q)$. The CCE solution is defined by 
\begin{equation}
\Phi^{CCE}(R,q)=\{ u(Z') \mid \exists \, p' : (Z',p') \in \CC(R,q)\}
\end{equation}

\subsection{The Naive Egalitarian per Object}
\label{sec:epo}
Finally, a naive and highly intuitive solution (that I use as a benchmark only) breaks up the allocation problem into $m$ sub-problems of assigning $q_k$ units of object $k$ into $R_{Nk}$, distributing an equal share of object $k$ among all agents who find it acceptable. I call this solution {\bf Egalitarian Per Object (EPO)}. Given an MAP $(R,q)$, the EPO solution assigns to each agent
\begin{equation}
\phi^{EPO}_i(R,q)= \sum_{k \in M} r_{ik} \cdot \frac{q_k}{\left\vert R_{Nk} \right\vert}
\end{equation}

In the dichotomous preference domain, EPO is equivalent to the well-known random priority mechanism, also known as random serial dictatorship.\footnote{EPO would not be efficient in a more general domain of preferences. The equivalence with random priority would also disappear.} I do not consider EPO to be an appropriate solution for MAPs because it ignores the interaction between the $m$ assignment problems corresponding to each object. EPO also fails the following basic fairness property: if $n-1$ agents receive at least one object, the $n$-th agent also receives at least one object; see Example 1 for an illustration.

One could also consider other solutions discussed in the literature, in particular the probabilistic serial rule, defined by \cite{bogomolnaia2001}. I do not consider this solution for two reasons. First, the probabilistic serial rule is appealing in scenarios where different notions of efficiency do not coincide. This is not the case for MAPs. Second, it was originally defined for assignment problems with strict preferences. Even though \cite{katta2006} extend the probabilistic serial rule to allow for indifferences, their extension is only defined for single-unit assignment problems.

\subsection{Two Examples Showing that All Solutions Are Disjoint}
\begin{example}[Multivalued CCE differs from EPO]Table \ref{tab:example1} shows the different outcomes that our three solutions produce for a problem with $n=6$, $m=3$, and $(R,q)$ given in subtable \ref{tab:r1.1}. The CCE utilities are written in brackets in subtable \ref{tab:r1.2} because there are CCE that support utility profiles between $(2.4,1.4,1)$ and $(2.25,2,1)$ with $0 \leq p_\gamma \leq \frac{4}{9}$. This multiplicity is interesting: the competitive solution is always unique in the corresponding utility profile in Fisher markets (\citealp{jain2010}, p.87). It is also problematic, as there is no obvious selection from the CCE.
\end{example}

\begin{table}[ht]
\centering
\caption{CCE is multi-valued.}
\label{tab:example1}
\begin{subtable}[t]{0.45\linewidth}
\centering
\begin{tabular}{|l|l|l|l|l|}
\hline
$N \backslash M$    & $\alpha$ & $\beta$ & $\gamma$ & Total  \\
\hline
$a:d$     & 1   & 1    & 1   & 3   \\
$e$	    & 1   & 1    & 0   & 2     \\
$f$	    & 1   & 0    & 0   & 1    \\
\hline
Total & 6   & 5    & 4   &    \\
\hline
\hline
$q$ & 4   & 4    & 4   &    \\
\hline
\end{tabular}
\caption{Corresponding $R$ matrix.}
\label{tab:r1.1}
\end{subtable}
\qquad 
\begin{subtable}[t]{0.45\linewidth}
\centering
\begin{tabular}{|l|l|l|l|}
\hline
$N$    		& ES & CCE & EPO \\
\hline
$a$:$d$ 		& 2.25 				& [2.25 - 2.4]  & 2.47	\\
$e$	    	& 2  					& [1.4 - 2]  		& 1.47     \\
$f$   			& 1						& 1    					& 0.67   \\
\hline
\end{tabular}
\caption{Utility profiles for each solution.}
\label{tab:r1.2}
\end{subtable}
\end{table}

Any CCE in Example 1 gives one unit of object $\alpha$ to agent $f$. This implies that there are no CCE prices that support the EPO outcome and thus is a strong argument against this solution, as competitive equilibria are considered ``{\it essentially the description of perfect justice}'' \citep{arnsperger1994}, and the base of Dworkin's ``{\it equality of resources}'' \citep{dworkin1981}. In consequence, the EPO solution is not ideal. But interestingly, the ES solution can also produce outcomes that cannot be supported as a CCE.

\begin{example}[ES differs from CCE] I show this using a MAP with $n=9$, $m=6$, and $(R,q)$ given in subtable \ref{tab:r2.1}. Note that in the single-unit case (Theorem 1 in BM04), the ES is always supported by competitive prices.
\end{example}

\begin{table}[ht]
\centering
\caption{ES and CCE are disjoint.}
\label{tab:example2}
\begin{subtable}[t]{0.45\linewidth}
\centering
\begin{tabular}{|l|l|l|l|l|l|}
\hline
$N \backslash M$    & $\alpha$ & $\beta$ & $\gamma,\delta$ & $\epsilon,\zeta$ &Total  \\
\hline
$a:c$     & 1   & 1    & 0     & 0      & 2   \\
$d$		  & 0   & 1    & 1     & 0        & 3     \\
$e$	    & 0   & 1    & 0     & 1        & 3   \\
$f:i$	    & 1   & 0    & 1     & 1      & 5     \\
\hline
Total & 7   & 5    & 5   & 5 & \\
\hline
\hline
$q$ & 4   & 4    & 4   & 4&  \\
\hline
\end{tabular}
\caption{Corresponding $R$ matrix.}
\label{tab:r2.1}
\end{subtable}
\quad 
\begin{subtable}[t]{0.45\linewidth}
\centering
\begin{tabular}{|l|l|l|l|l|l|}
\hline
$\alpha$ & $\beta$ & $\gamma,\delta$ & $\epsilon,\zeta$ &Total \\
\hline
1   & 0.97    &  0   &  0   & 1.97     \\
0   & 0.54    & 1     &  0    & 2.54     \\
0   & 0.54    & 0      &  1    & 2.54     \\
0.25   & 0    & 0.75   & 0.75    & 3.25     \\
\hline
\end{tabular}
\caption{Corresponding $Z^*$ (CCE).}
\label{tab:r2.2}
\end{subtable}
\end{table}

If the ES solution (2, 2.5, 2.5, 3.25) could be supported as a CCE, then $p_\alpha=p_\gamma=p_\delta=p_\epsilon=p_\zeta$ because agents $f$:$i$ obtain those objects with positive probability but do not exhaust them. Furthermore, agents $d$:$i$ must spend their whole budget, implying prices $p_\alpha=\frac{4}{13}$ and $p_\beta=\frac{10}{13}$. However, at such prices, the ES utility for agents $a$:$c$ is unaffordable $(\frac{14}{13} > 1)$. 

The fact that ES and CCE do not coincide is interesting: in the non constrained context, the competitive solution can be computed by maximizing the Nash product, solving what is known as the Eisenberg-Gale program \cite[][see chapter 7 in \citealp{moulin2003} for a textbook treatment]{eisenberg1961,eisenberg1959,chipman1974}. That the competitive solution cannot be computed solving the Eisenberg-Gale program implies that we lack an algorithm for computing the competitive equilibrium. The Eisenberg-Gale program is otherwise a rather robust result since it extends to a large family of utility functions beyond the linear case \citep{jain2010}, as well as to the mixed division of objects and bads \citep{bogomolnaia2017}. 

The multiplicity of the competitive solution and its non-equivalence with the egalitarian outcome justify the new terminology of CCE. For any MAP, the set of CCE is non-empty, a result that I prove in the Appendix using a standard fixed point argument. I summarize these findings in Theorem 2.

\begin{theorem}
\label{theorem:existence}
The ES solution is well-defined and single-valued, and the CCE solution exists. Their intersection can be empty.
\end{theorem}

\subsection{Envy}
It is easy to see that both the ES and CCE solutions are envy-free. A solution $\phi$ is {\bf envy-free} if, for any MAP $(R,q)$ with agents $i$ and $j$ such that $R_{iM} \subseteq R_{jM}$, $\phi_i(R,q) \leq \phi_j(R,q)$. For the multi-valued CCE, envy-freeness holds for any selection from it.

\begin{lemma}
\label{lemma:envyfree}
ES and CCE are envy-free.
\end{lemma}

I postpone an easy proof. Note that there is no efficient solution that is strongly envy-free, i.e. that for any MAP $(R,q)$ with agents $i$ and $j$ such that $\left\vert R_{iM} \right\vert < \left\vert R_{jM} \right\vert$, $\phi_i(R,q) \leq \phi_j(R,q)$, see \cite{ortega2016}.

\section{Manipulation by Groups}
\label{sec:manipulation}
I consider agents' manipulation in the direct revelation mechanism associated with each solution. For this purpose, we need to know exactly how objects are assigned and not just the total utility that each agent receives. A {\bf detailed solution} $\psi$ maps every MAP ($R,q$) into an RAM $Z \in \FF(R,q)$, specifying which share of each object is allocated to each agent, whereas a welfarist solution $\phi$ maps every MAP into a utility profile $U \in \UU(R,q)$ and only tells us the expected number of objects received by each agent. Every detailed solution $\psi$ projects onto the welfarist solution $\phi(R,q)=u(\psi(R,q))$. The direct revelation mechanism associated with a detailed solution $\psi$ is such that all agents reveal their preferences $R_{iM}$, and then $\psi$ is applied to the corresponding MAP $(R,q)$, implementing the RAM $\psi(R,q)=Z$.

I assume that agent $i$ with the true preferences $R_{iM}$ can only misrepresent her preferences by understating the number of objects that she finds acceptable, i.e. by declaring a preference profile $R'_{iM} \subset R_{iM}$ (we then say that $R'_{iM}$ is IR for $R_{iM}$). I use this assumption for two reasons. The first is theoretical: I have not specified the dis-utility that the consumption of an undesirable object brings to an agent, as I have only focused on individually rational assignments. I would need to specify such dis-utility to analyse the manipulation of a solution by exaggerating the set of acceptable objects. The second reason is that such an assumption has already been imposed in the study of scheduling problems (e.g. \citealp{koutsoupias2014}). In many scheduling problems motivating MAPs, cancelling consumption could be strongly punished by the central clearinghouse, particularly when other agents' consumption depends on other agents fully exhausting their bundles (no double tennis match can be made with only 3 out of 4 players).\footnote{If agents are allowed to report undesirable objects as desirable, the egalitarian and even priority solutions are manipulable by groups (an agent with a higher priority requests an object she does not want and gives it to a member of the manipulating coalition).}$\textcolor{blue}{^,}$\footnote{BM04 impose an equivalent assumption: they require that the RAM of every MAP must be individually rational according to the agents' true preferences.}

A detailed solution $\psi$ is {\bf group strategy-proof} if for every MAP $(R,q)$ and every coalition $S \subset N$, $\nexists$ $R'$ satisfying i) $R'_{jM}=R_{jM}$ $\forall j \notin S $, and ii) $R'_{SM}$ is IR for $R_{SM}$, such that
\begin{eqnarray}
\forall i \in S, \quad u_i(\psi(R',q)) \geq u_i (\psi(R,q)) 
\end{eqnarray}

with strict inequality for at least one agent in $S$. A welfarist solution $\phi$ is {\bf group strategy-proof} if every detailed solution $\psi$ projecting onto $\phi$ is group strategy-proof.

BM04 show that no deterministic solution is group strategy-proof when agents can obtain at most one object. Deterministic solutions include priority ones, i.e. those in which agents choose sequentially their most preferred available bundle according to some pre-specified order. The reason of why deterministic solutions are manipulable by groups is that the agent with the highest priority could change her report and still receive one acceptable alternative, leaving her utility unchanged and, at the same time, benefiting an agent with low priority: a property known as bossiness. 

This argument does not extend to MAPs. Because agents can obtain multiple objects, the agent with higher priority can belong to a manipulating coalition only by claiming fewer objects. But since she has the highest priority, it is immediate that such manipulation would always give her strictly less utility, so she cannot be in the coalition. The same argument applies to all remaining agents and, consequently,
\begin{lemma}
Any deterministic priority solution is group strategy-proof.
\end{lemma}

The previous Lemma shows that group strategy-proofness is relatively easy to achieve for MAPs in the dichotomous domain. In fact, I show below that the ES solution is also group strategy-proof. Is CCE also group strategy-proof? There are two extensions of our group strategy-proofness definition to set valued solutions. 

The first requires that for every MAP ($R,q$), there is no competitive equilibrium of the manipulated MAP ($R',q$) that is weakly better than every competitive equilibria of the original problem $(R,q)$, for every member of the manipulating coalition $S$. A stronger extension is that there is at least one competitive equilibrium of $(R,q)$ which yields a weakly higher utility than some competitive equilibrium of $(R',q)$, with strict inequality for at least one member of the deviating coalition $S$. It turns out that CCE violates both conditions. The reason for this is that a group can coordinate to make several objects perfect, thus allowing those objects to have a zero price.

\begin{theorem}
	\label{theorem:manipulation}
ES is group strategy-proof but CCE is not.
\end{theorem}

The proof of ES being group strategy-proof can be found in the Appendix, but I show that CCE is unambiguously manipulable by groups below.

\begin{example}[CCE not group strategy-proof]Let $n=7$, $m=4$, and $(R,q)$ given by Table \ref{tab:example3}. 

\begin{table}[ht]
\centering
\caption{Example 3.}
\label{tab:example3}
\begin{subtable}[t]{0.45\linewidth}
\centering
\begin{tabular}{|l|l|l|l|l|l|l|}
\hline
$N \backslash M$  & $\alpha$ & $\beta$ &  $\gamma$  & $\delta$     & $\Phi^{CCE}$\\
\hline
$\pmb a$&1 & \bf 1 & 1 & 1 & \bf 2.5\\
$\pmb b$&1 & 1 & \bf 1 & 1 & \bf 2.5\\
$\pmb c$&1 & 1 & 1 & \bf 1 & \bf 2.5\\
$d$&1 & 0 & 1 & 1 & 2.5\\
$e$&1 & 1 & 0 & 1 & 2.5\\
$f$&1 & 1 & 1 & 0 & 2.5\\
$g$&1 & 0 & 0 & 0 & 1\\
\hline
Total &7 & 5 & 5 & 5	& \\
\hline
\hline
$q$ &4 & 4 & 4 & 4	& \\
\hline
\end{tabular}
\caption{True preferences $R$.}
\label{tab:71}
\end{subtable}
\qquad 
\begin{subtable}[t]{0.45\linewidth}
\centering
\begin{tabular}{|l|l|l|l|l|l|}
\hline
$\alpha$ & $\beta$ &  $\gamma$  & $\delta$     & $\Phi^{CCE}$\\
\hline
1 & \bf 0 & 1 & 1 & \bf [2.5 - 2.57]\\
1 & 1 & \bf 0 & 1 & \bf [2.5 - 2.57]\\
1 & 1 & 0 & \bf 0 & \bf [2.5 - 2.57]\\
1 & 0 & 1 & 1 & [2.5 - 2.57]\\
1 & 1 & 0 & 1 &  [2.5 - 2.57]\\
1 & 1 & 1 & 0 &  [2.5 - 2.57]\\
1 & 0 & 0 &  0 & [0.57 - 1]\\
\hline
Total &7 & 4 & 4 &  \\
\hline
\hline
$q$ &4 & 4 & 4 &  \\
\hline

\end{tabular}
\caption{Misreport $R'$ for $S=\{a,b,c\}$.}
\label{tab:r72}
\end{subtable}
\end{table}
\end{example}

Consider the coalition $S=\{a,b,c\}$. When agents submit their real preferences, there exists a unique CCE that supports the ES solution: agents $a,b,$ and $c$ obtain 2.5 expected objects. By changing their report each for a different object, as in subtable \ref{tab:r72}, they make objects $\beta$, $\gamma$ and $\delta$ perfect, consequently enlarging the set of CCE solutions, which includes utilities that are always weakly above 2.5 and up to 2.57. By misrepresenting and creating artificially perfect objects, they allow those to be priced at 0, weakly increasing the number of expected objects received in any competitive equilibria of ($R',q$), at the expense of agents with limited acceptable objects, in this case $g$.

I do not discuss strategy-proofness (manipulation by individuals on their own) since it is immediate that ES and CCE (and EPO) are strategy-proof. For CCE, we can construct a selection of it that is strategy-proof, since reducing the total demand for an object either reduces its price, relatively increasing the price of other objects, or leaves its price unchanged.

Efficiency, fairness, and non-manipulability are standard goals in the design of resource allocation mechanisms. Before concluding, I discuss a new goal that arises naturally for MAPs.

\section{Independence of Perfect Objects}
\label{sec:perfect}

Some solutions do not depend on the number of perfect objects desired by an agent. If an agent finds a new perfect object to be acceptable, we could expect that she would always receive one extra expected unit. This is what our following property captures.

A solution $\phi$ is {\bf independent of perfect objects (IPO)} if, for every MAP, every $i \in N$ and for any of its perfect extensions $([R \; R_{Nk'}], \overline q)$,
\begin{equation}
\phi_i(R,q)+1=\phi_i([R \; R_{Nk'}],  \overline q)
\label{eq:wealth}
\end{equation}

IPO is a desirable property for two reasons. First, perfect objects belong unambiguously to agents who find them acceptable, so they can argue that they should obtain them fully, irrespective of the share they obtain from over-demanded objects. Second, if the clearinghouse uses a solution that was not IPO, the set of agents who find perfect objects acceptable could avoid reporting their demand for perfect objects and obtain them fully outside the centralized mechanism, a real concern for scheduling applications in which agents may organize teamwork activities on their own.

CCE (partially) satisfies this requirement.\footnote{EPO also satisfies IPO. Once more, EPO performs very poorly with respect to fairness considerations so I do not analyse it further.}

\begin{lemma}
\label{lemma:ipo}
Although ES is not IPO, there exists a selection of CCE that satisfies IPO. 
\end{lemma}

Lemma \ref{lemma:ipo} highlights that CCE can always assign a zero price to all perfect objects: this is how we construct the selection of CCE that satisfies IPO. But it may also assign a zero price to some perfect objects only, or to no perfect object at all. The designer has a high degree of flexibility in choosing the equilibrium prices.

The selection problem extends to \citeauthor{budish2011}\textcolor{blue}{'s} (\citeyear{budish2011}) competitive mechanism for CAP in which students reveal their preferences to a centralized clearinghouse which announces a corresponding equilibrium allocation. Budish argues that this mechanism is transparent, meaning that students can verify that the allocation is an equilibrium. But the mechanism can be ``manipulated from the inside'', selectively assigning zero prices to hand-picked courses, while at the same time rightly arguing that it produces a competitive allocation.

If IPO must be achieved (a decision depending on the context and the designer's objectives), we would like to have a solution that, at the same time, avoids the multiplicity problem of the CCE, while being envy-free and as fair as possible. Such a solution exists: we call it the {\bf refined egalitarian solution (ES*)}. To define it, we use the partition of $M$ into $\PP(R,q)$ and $\OO(R,q)$, and split the original MAP $(R,q)$ into two independent problems $(R_{N\PP(R,q)},q_{\PP(R,q)})$ and $(R_{N\OO(R,q)},q_{\OO(R,q)})$, which correspond to the independent MAPs with perfect and over-demanded objects, respectively. ES* is given by

\begin{equation}
\phi^{\es^*}_i(R,q)=\phi^\es(R_{N\OO(R,q)},q_{\OO(R,q)}) + \left\vert R_{i\, \PP(R,q)} \right\vert
\end{equation}

ES* takes the egalitarian solution for the MAP with over-demanded objects only, and adds the number of perfect objects in which a player is available. ES* is close to a suggestion in \cite{budish2011}. Noting that some courses may be in excess supply, he proposes to run the allocation mechanism only on the set of over-demanded courses: ``{\it if some courses are known to be in substantial excess supply, we can reformulate the problem as one of allocating only the potential scarce courses}''. ES* formalizes this suggestion. It also satisfies several desiderata.

\begin{lemma}
\label{thm:refined}
The ES* solution is well-defined and single-valued, efficient, IPO, envy-free, and Lorenz dominant for the problem $(R_{N\OO(R,q)},q)$.
\end{lemma}

It is immediate that ES* is single-valued, efficient and IPO. The remaining properties are straightforward modifications of the proofs of Lemmas 1 and 2 and Theorem 1. Unfortunately, the properties in Lemma \ref{thm:refined} come at a cost: ES* is not group strategy-proof.\footnote{For an example, use the MAP and manipulation $R'$ illustrated in Example \ref{tab:example3}.} ES* can be manipulated by groups reducing their availability in order to make some objects perfect. Therefore, the members of the manipulating coalition obtain those objects fully, while also obtaining an egalitarian fraction of the remaining over-demanded problem.

Group strategy-proofness and IPO are compatible. EPO satisfy them both (plus envy-freeness and Pareto efficiency). However, its poor performance with respect to fairness makes it inappropriate for the problems I have considered in this paper, as argued in subsection \ref{sec:epo}.

\section{Conclusion}
\label{sec:conclusion}
For multi-unit assignment problems in the dichotomous preference domain, the egalitarian solution is Lorenz dominant, single-valued and group strategy-proof. For these reasons, I recommend its use as a solution instead of the celebrated competitive equilibrium with equal incomes, which fails these three desirable properties. If the market designer is interested in satisfying the property of independence of perfect objects, the refined egalitarian solution becomes an appealing alternative.

\pagebreak
\section*{Appendix: Proofs}
\noindent {\bf Theorem 1} {\it The ES solution is Lorenz dominant in the set of efficient utility profiles.}
\begin{proof}
	Fix a MAP $(R,q)$. Consider the concave cooperative game $(N, \mu)$ where $\mu:2^N \to \RR$ is a function that assigns, to each subset of agents, the maximum number of objects that they can obtain together. To formalize this intuitive function, given a coalition $S \subset N$, let us partition the set of objects $M$ into $M^+(S)$ and $M^-(S)$, defined as
	\begin{equation}
	M^+(S)=\{k \in M : \left\vert R_{Sk} \right\vert < q_k\}
	\end{equation}
	
	The function $\mu$ is given by
	\begin{equation}
	\mu(S)= \sum_{k \in M^+(S)} \sum_{i \in S}  r_{ik} + \sum_{k \in M^-(S)} q_k
	\end{equation}
	
	This function is clearly submodular, i.e. for any two subsets $T,S \subset N$
	\begin{equation}
	\mu(S) + \mu(T) \geq \mu(S \cup T) + \mu(S \cap T)
	\end{equation}
	
	The ``core from above'' is defined as the following set of profiles
	\begin{equation}
	C(R,q)=\{x \in \RR^n \mid \sum_{i \in N} x_i = \nu (R,q) \text{ and } \nexists \, S\subset N : \sum_{i \in S} x_i  > \mu(S)\}
	\end{equation}
	
	It follows from Theorem 3 in \cite{dutta1989} that the set $C(R,q)$ has a Lorenz dominant element and is the egalitarian solution. But by construction of the ``core from above'', $\UU(R,q) \subset C(R,q)$, the ES solution is also Lorenz dominant in the set of efficient utility profiles $\UU(R,q)$. 
\end{proof}

\noindent {\bf Theorem 2} {\it For generalized tennis problems, the ES solution is well-defined and single-valued, and the CCE solution exists. Their intersection can be empty.}
\begin{proof}
	Fix a MAP $(R,q)$. Let $p \in \RR^m_+$ be an arbitrary price vector such that $p \cdot c = n$, and use the notation $y_i= R_{iM}$ to denote the optimal consumption bundle for agent $i \in N$, and $y_N=(\left\vert R_{N1} \right\vert, \ldots,\left\vert R_{Nm} \right\vert)$. Note that
	\begin{equation}
	p \cdot y_N \geq p \cdot q
	\end{equation}
	
	Define the vector $\vec{\lambda}$ as
	\begin{equation}
	\vec{\lambda}(p)=(\lambda_1, \ldots, \lambda_n)=\text{UNIF}\{p \cdot y_i; n\}
	\end{equation}
	
	where UNIF denotes the uniform rationing rule: a mapping that gives to every agent the money needed to buy her preferred bundle of objects as long as it is less than $\lambda$, chosen so that $p \cdot \vec{\lambda}=n$. Define the sets of satiated and non-satiated agents
	\begin{eqnarray}
	N_0(p)&=&\{ i \in N \mid \lambda_i = p \cdot y_i\}\\
	N_+(p)&=&\{ i \in N \mid \lambda_i < p \cdot y_i\} 
	\end{eqnarray}
	
	So that $\lambda_i =\lambda \, \forall i \in N_+$. Define the demand correspondence $d_i(p)$ as
	\begin{equation}
	d_i(p) = \arg\max_{Z_{iM} \in \IR(R_{iM})}\{p \cdot Z_{iM} \leq \lambda_i \}
	\end{equation}
	
	where $\IR(R_{iM})$ denotes the set of individually rational assignments for $R_{iM}$. Note that $d_i(p)=\{y_i$\} for every $i \in N_0(p)$, while for agents in $N_+(p)$, any vector $z_i \in d_i(p)$ satisfies $p \cdot z_i = \lambda$. By Berge's maximum theorem, the demand correspondence is upper hemi-continuous and convex valued. The excess demand correspondence for the whole society, which inherits the properties of $d_i$, is given by
	\begin{equation}
	e(p)=d_N(p) - q
	\end{equation}
	
	where $d_N(p)$ denotes the aggregate demand correspondence for each object. Using the Gale-Nikaido-Debreu theorem (Theorem 7 in pp. 716-718 of \cite{debreu1982}), we know that there exists both a price vector $p^* \in R_+$ and an excess demand vector $x^* \in e(p^*)$ for which the following two conditions are satisfied
	\begin{eqnarray}
	\label{eq:feasi}x^*& = &\vec{0}\\
	\label{eq:walras}p^* \cdot x^*&=&0
	\end{eqnarray}
	
	Where Walras' law in equation (\ref{eq:walras}) holds by construction of $\vec{\lambda}$ and $d$. Finally, $\forall i \in N $
	\begin{equation}
	Z^*_{iM}=d_i(p^*)
	\end{equation}
	
	so that the corresponding $Z^* \in \FF(R,q)$  by equation (\ref{eq:feasi}), concluding the proof of existence of CCE. That ES is single-valued follows from Theorem 1. I have shown in Example \ref{tab:example3} that for some MAP there do not exist prices that support the ES as a CCE.
\end{proof}

\noindent {\bf Lemma 2} {\it ES and CCE are envy-free.}
\begin{proof}
	For an arbitrary MAP, let $\phi^\es(R,q)=(U_1,\ldots, U_i, U_j, \ldots, U_n)$, and assume agent $i$ is envious of $j$, which means that $R_{jM} \subseteq R_{iM}$ and that there exists a Pigou-Dalton transfer $\epsilon$ so that the utility profile $U'=(U_1,\ldots, U_i+\epsilon, U_j-\epsilon, \ldots, U_n) \in \UU(R,q)$. But $U'$ Lorenz dominates $\phi^\es(R,q)$, so $\phi^\es(R,q)$ was not the ES solution, a contradiction.
	
	Any selection of the CCE solution is envy-free because of the standard argument: if there is any agent who is envious, she could afford the schedule of the agent she envies.
\end{proof}

\noindent {\bf Theorem 3} {\it ES is group strategy-proof but CCE is not.}

\noindent I have shown that CCE is not group strategy-proof in the main text. To show that ES is group strategy-proof, I start with a few preliminaries. Let $\ZZ$ denote the set of all feasible RAMs supporting the egalitarian solution, i.e.
\begin{equation}
\ZZ=\{ Z \in \FF(R,q) \mid \forall i \in N : \sum_{k \in M} z_{ik}=\phi^\es_i (R,q) \}
\end{equation} 

A rule is non-bossy if no agent can affect someone else's allocation without changing her own utility. That is, a solution $\phi$ is {\bf non-bossy} if, for every MAP ($R,q$), $\forall i \in N$, and any manipulation $R'$ such that 1) $\forall j \neq i, R_{jM}=R'_{jM}$, and 2) $R'_{iM} \subsetneq R_{iM}$, we have
\begin{equation}
\phi_i(R,q)=\phi_i(R',q) \quad \textrm{ only if } \quad  \phi(R,q)=\phi(R',q)
\end{equation}

We prove a useful auxiliary Lemma below.

\begin{lemma}
	ES is non-bossy.
\end{lemma}

\begin{proof}
	We proceed by way of contradiction. Let $R'$ be as specified in the previous definition. The manipulation may come from a reduction of availability in three types of objects:
	
	1. $k \in \OO(R,q)$ and $\{k \in M \mid \exists Z \in \ZZ : z_{ik}=0\}$, and hence there is a way to implement the ES solution even when agent $i$ misreported, so her change in availability is inconsequential and all utilities remain the same, so agent $i$ cannot be bossy.
	
	2. $k \in \OO(R,q)$ and $\{k \in M \mid \forall Z \in \ZZ : z_{ik}>0\}$, so clearly agent $i$'s
	utility changes, so she cannot be bossy. 
	
	31. $k \in \PP(R,q)$, but if agent $i$ reduces the number of perfect goods, she always reduces the utility she obtains (as I prove below), so her utility is not constant and she cannot be bossy.
	
	Now I prove that reducing the number of perfect objects in which agent $i$ is available always strictly reduces her utility. The certain loss of the perfect object(s) must be exactly compensated by an increase of the shares she gets from all over-demanded objects, which is constant in any $Z \in \ZZ$. Agent $i$ was not getting full shares of those objects (as otherwise we obtain a contradiction) so another agent(s) $j$ must be obtaining shares for those objects, implying $\phi_j^\es(R,q) \leq \phi_i^\es(R,q)$ (because otherwise the ES would give those shares to agent $i$). Some of the shares obtained by agent $j$ in $\phi(R,q)$ must be transferred to agent $i$ in $\phi(R',q)$: this is a Pigou-Dalton transfer because if agent $i$ did not obtain a lower utility in the misrepresented problem then he would not obtain the shares of $j$. Moreover, 
	\begin{equation}
	\phi_i^\es(R,q) - 1 < \phi_j^\es(R,q) \leq \phi_i^\es(R,q)
	\end{equation}
	
	as otherwise $j$ does not transfer any shares to $i$ when $i$ reduces the number of perfect objects. Let $\gamma$ be the Pigou-Dalton transfer from $j$ to $i$ required so that the utility of $i$ is kept constant. We have
	\begin{equation}
	\phi_i^\es(R',q)=\phi_i^\es(R,q) - 1 + \gamma = \phi_j^\es(R,q)  - \gamma < \phi_i^\es(R,j)
	\end{equation}
	
	showing that indeed reducing the number of perfect objects always yields lower utility, and thus concluding the proof that ES is non-bossy.\end{proof}

We are now ready to prove that ES is group strategy-proof. We will do this by showing that nobody can join a manipulating coalition.

\begin{proof}
	By way of contradiction, assume there exists a MAP $(R,q)$, a coalition $S \subsetneq N$, and a manipulation $R'$ such that, for all $i \in S$ $\phi^\es_i(R',q) \geq \phi^\es_i(R,q)$, and for some $j \in S$ $\phi^\es_j(R',q) > \phi^\es_j(R,q)$. 
	
	Let $\phi^\es(R,q)=U^\es$ and order the agents such that $U^\es_1 \leq \ldots \leq U^\es_n$. We will show by induction on the order of agents the following property
	\begin{equation}
	\label{eq:notins}
	i \notin S
	\end{equation}
	
	There are two cases in which an agent $i$ can be in $S$. Case 1) either she gets more utility, $\phi_i^\es(R',q) > \phi_i^\es(R,q)$, or case 2) she gets the same utility but she changes her reported preferences to help another member of $S$. This is ruled out by the non-bossiness of ES so we focus on case 1) only. 
	
	We prove it for $i=1$ first, i.e. the agent with lowest utility. Agent 1 gets a strictly higher number of objects with the new profile $R'$, which must come from a set of objects $K \subseteq \OO(R,q)$ from which he was not getting full shares ($K=\{k \in M \mid \exists Z \in \ZZ : 0 < z_{ik}<1\}$), for which agents $2, \ldots, t$ are also available and $U^\es_1=U^\es_2=\ldots=U^\es_t$. Those agents exhaust $q_k$ entirely; i.e. $\forall k \in K, \forall Z \in \ZZ, \sum_1^t z_{ik}=q_k$. 
	
	Let $T=\{1,\ldots,t\}\cap S$. For any preference matrix $R'_{TM}$ that is individually rational for $R_{TM}$, the objects $\{k \in K \mid R_{Nk} \neq R'_{Nk}\}$ become less over-demanded for agents $\{1,\ldots,t\}\setminus T$, and therefore the agents in $T$ get less objects as a whole. Therefore there must be at least one agent in $T$ who is worst off, and the coalition $S$ is not viable. Therefore $1 \notin S$.
	
	Now we assume that $i \notin S$ for agent $i=h-1$ and we show it holds for agent $h$. We must have that $U^\es_h<\left\vert R_{hM} \right\vert$. We assume $\phi^\es_1(R,q)_1^\es< \phi_h^\es(R,q)$ as otherwise our argument for agent $1$ works exactly the same.
	
	If agent $h \in S$, it must be that there exists a manipulation $R'$ so that $\phi_h(R',q)> \phi_h(R,q)$. The increase in her utility must come from more object shares on over-demanded objects which she was not obtaining fully, i.e. $K^h=\{k \in M \mid \exists Z \in \ZZ : 0<z_{hk}<1\}$. Some of these objects are exhausted by agents $1, \ldots, h-1$. There is no way agent $h$ could get more shares from any of those objects because $\{1, \ldots, h-1\} \cap S = \emptyset$ by our induction step.
	
	Therefore, the increase must come from objects that are not exhausted by $\{1,\ldots,h-1\}$. Those objects become less over-demanded for $\{h, \ldots, n\} \setminus S$, and therefore agents in $S$ get less object shares as a whole. It follows that there must be a agent in $S$ who gets less utility, so coalition $S$ is not viable. Therefore $h \notin S$, and this concludes the proof.
\end{proof}

\noindent {\bf Lemma 4} {\it Although ES is not IPO, there exists a selection of CCE that satisfies IPO.}
\begin{proof}
	It is straightforward to show that ES is not IPO. Let $n=5, M=\{\alpha\}, q=4$, and $R^\top=[1 \; 1\; 1\; 1\; 1]$. $\phi^\es_i(R,q)=0.8$ for any agent, but adding a perfect object $k'$ with capacity 4 for any agent $i$ changes $\phi^\es_i([R \, R_{Nk'}], (4,4))=1.75 \neq 2$. 
	
	To show that there is a selection of $\Phi^{CCE}$ that is IPO, let $(Z^*,p^*)$ be a CCE of $(R,q)$ and $([R \, R_{Nk'}], \overline q)$ be a perfect extension of $(R,q)$ . Then fix $p^*_{k'}=0$ and, for every $i \in N$ let $z^*_{ik'}=1$ if $r_{ik'}=1$, and 0 otherwise. The pair $([Z^* \, Z^*_{Nk'}],(p^*_1, \ldots, p^*_n, 0))$ is a CCE of the perfect extension $([R \, R_{Nk'}],\overline q)$, because everybody interested in the perfect object is able to afford it, and the demand for $k'$ equals its supply, because the new object $k'$ is perfect.
\end{proof}

\section*{References}

\setlength{\bibsep}{0cm}

\end{document}